\documentclass{article}

\usepackage{amsfonts}

\usepackage{amsmath}
\usepackage{amsmath,amssymb}
\usepackage{graphicx}

\newtheorem{theorem}{Theorem}

\newtheorem{definition}[theorem]{Definition}

\newtheorem{lemma}[theorem]{Lemma}

\newenvironment{proof}[1][Proof]{\textbf{#1.} }{\ \rule{0.5em}{0.5em}}

\begin{document}

\title{On Structure Preserving Transformations of the It\={o} Generator Matrix for Model 
Reduction of Quantum Feedback Networks}
\author{Hendra I. Nurdin\thanks{%
H. I. Nurdin was with the Research School of Engineering, The Australian National University, when this work was completed.
He is now with the School of Electrical Engineering and Telecommunications, The University of New South Wales,
Sydney NSW 2052, Australia. Email: h.nurdin@unsw.edu.au.}~~and~~John E. Gough%
\thanks{%
Institute for Mathematics and Physics, Aberystwyth University, SY23 3BZ,
Wales, United Kingdom. Email: jug@aber.ac.uk.}}
\maketitle

\begin{abstract}

Two standard operations of model reduction for quantum feedback networks,
 internal connection elimination under the instantaneous feedback limit and
adiabatic elimination of fast degrees of freedom, are cast as structure preserving
transformations of It\={o} generator matrices. It is shown that the order in which they are applied is inconsequential. \end{abstract}

\section{Introduction}

\label{sec:intro}

The last two decades have seen the emergence and explosion of global
research activities in quantum information science that promise to deliver
quantum technologies, a class of technologies that rely on and exploit the
laws of quantum mechanics, which can beat the best known capabilities of
current technological systems in sensing, communication and computation.
Most of the envisioned quantum technologies are quantum information
processing systems that process quantum information \cite{NC00,DM03}.
Typical proposals are realized as quantum networks: linear quantum optical
computing \cite{KLM01}, the quantum internet \cite{Kimb08}, and quantum
error correction \cite{KNPM10,KPCM11}. Quantum networks have also been
experimentally realized in proof-of-principle demonstrations of quantum
information processing, see, e.g., \cite{YAF04,LSHNLCA10}. Besides quantum
information processing, quantum networks have also been proposed for new ultra
low power photonic devices that perform classical information processing. In particular,
photonic devices that act as photonic analogues of classical electronic
circuits and logic devices, e.g., \cite{NJD08,Nurd10a,Nurd10b,Mab11}.

Even relatively simple quantum networks may be difficult to simulate due to
the large number of variables that need to be propagated. It is therefore
necessary to look at model reduction. For instance, this has been used to obtain a
tractable network model of a coherent-feedback system implementing a quantum
error correction scheme for quantum memory \cite{KNPM10}. 
In particular, this involved reduced QSDE models for several
components that make up the nodes of the network. In fact, the process led
to a simple and intuitive quantum master equation that describes the
evolution of the composite state of the three qubits of the quantum memory
and the two atom-based optical switches which jointly act as a
coherent-feedback controller. The idea for this coherent-feedback
realization of a three qubit bit(phase)-flip quantum error correction code,
which can correct only for single qubit bit(phase)-flip errors, was
subsequently extended to a coherent-feedback realization of a nine qubit
Bacon-Shor subsystem code that can correct for arbitrary single qubit
errors \protect\cite{KPCM11}, see Figure \ref{fig:QEC3}. Again, here QSDE model reduction played a
crucial role in justifying the intuitive quantum master equation that
describes the operation of the coherent-feedback QEC circuit.

\begin{figure*}[htbp]
\centering
\includegraphics[scale=0.35]{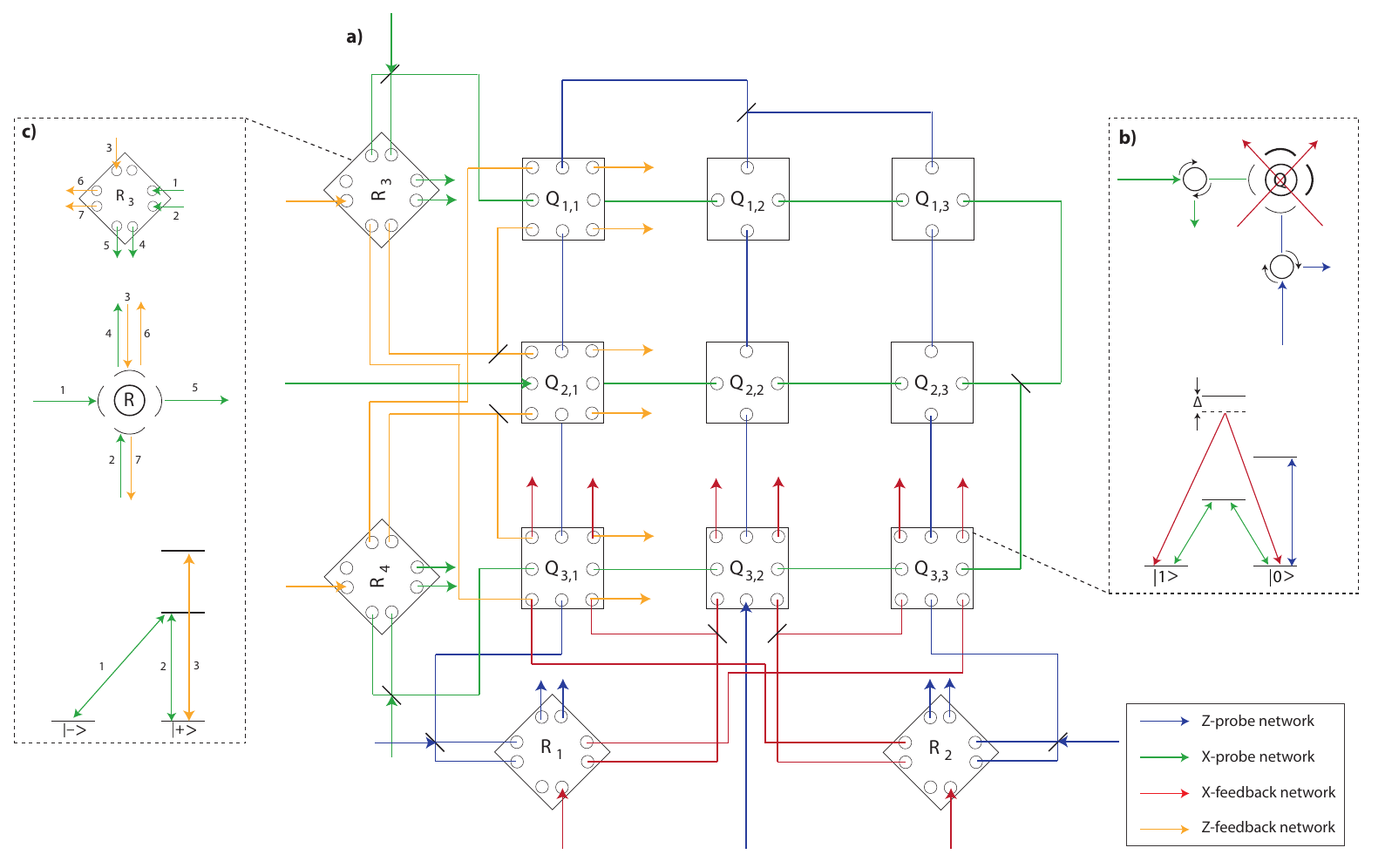}
\includegraphics[scale=0.30]{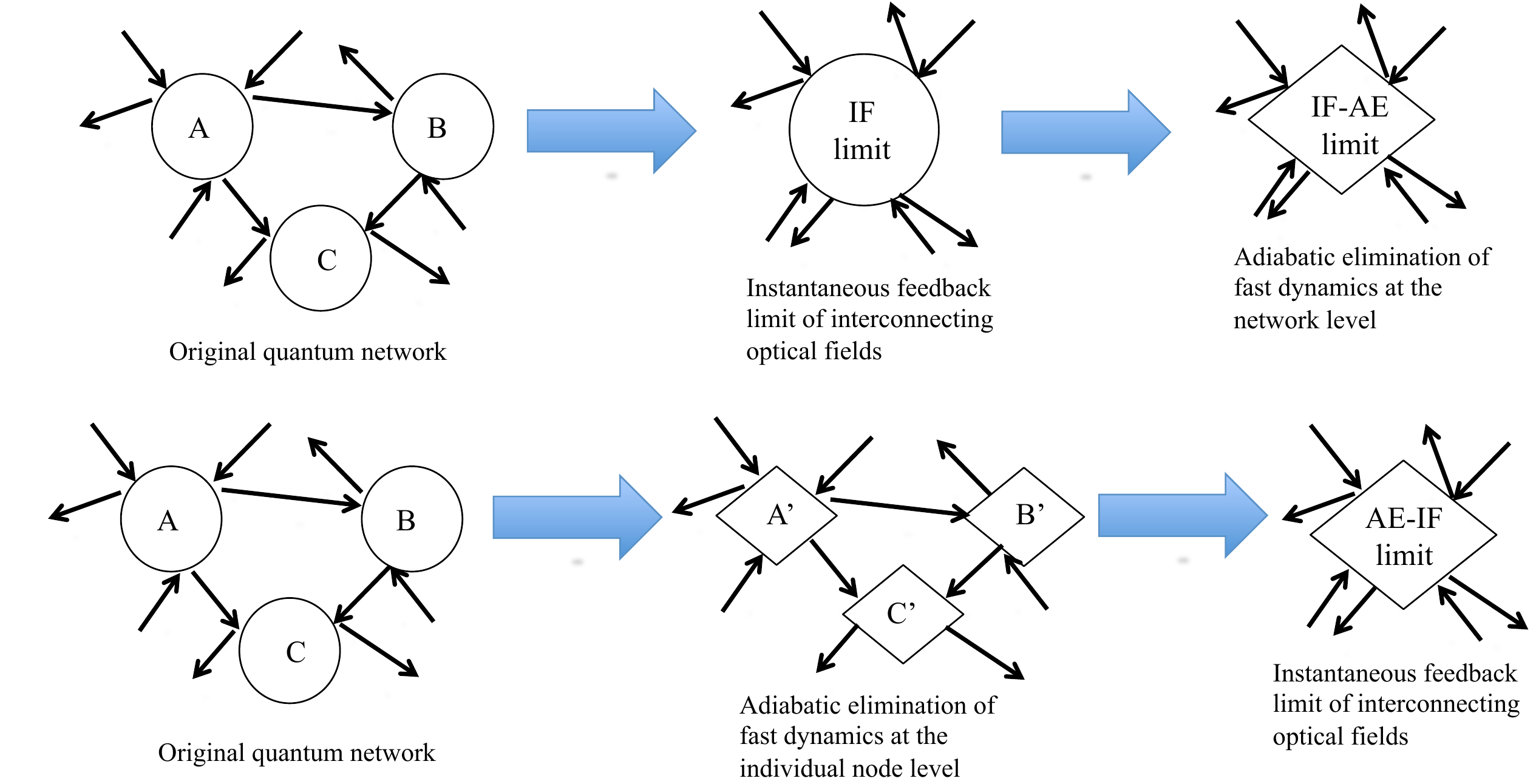}
\caption{\textbf{Left}: A coherent-feedback quantum network that implements a nine qubit
Bacon-Shor subsystem quantum error correction code from \protect\cite
{KPCM11}. The four relays $R_{1},R_{2},R_{3},R_{4}$ act jointly as a
coherent-feedback controller. \textbf{Top right}: complexity reduction of a quantum network by the instantaneous
feedback limit operation (IF) followed by the adiabatic elimination
operation (AE). \textbf{Bottom right}: complexity reduction of a quantum network by
the adiabatic elimination operation followed by the instantaneous feedback limit
operation. A circle denotes a node or quantum network before adiabatic
elimination while a rhombus denotes a node or quantum network after
adiabatic elimination. }
\label{fig:QEC3}
\end{figure*}

This paper considers a class of dynamical quantum networks with open Markov
quantum systems as nodes and in which nodes are interconnected by bosonic
optical fields (such as coherent laser beams). Here the optical fields serve
as quantum links or ``wires'' between nodes in the network. Time delays in
the propagation of the optical fields mean that the network as a whole is no
longer Markov, but fortunately, an \emph{effective Markov model} may be
recovered in the zero time delay limit \cite{CWG93,HJC93,GJ08,GJ07}. The
effective Markov model can then be viewed as a large single node network, as
illustrated in Fig.~\ref{fig:QEC3}. This kind of limit will be referred to
as an \emph{instantaneous feedback limit}.

Another commonly employed approximation is \emph{adiabatic elimination} (or 
\emph{singular perturbation}) of quantum systems that have fast and slow
sub-dynamics with well-separated time scales \cite{DK04,DJ99,BvHS07}.
Besides model simplification, adiabatic elimination has also proved to be a
powerful tool for the approximate engineering of ``exotic'' two or more body
couplings, see, e.g., \cite{WM93,NJD08,KBSM09,KNPM10}.

In \cite{GNW10} it was established, for a special class of quantum
networks containing fast oscillating quantum harmonic oscillators, that the
instantaneous feedback and adiabatic elimination limits are
interchangeable. The main contribution of the present paper is to extend the
results of \cite{GNW10} to general classes of quantum networks with 
Markovian components.

\section{Quantum stochastic differential equations and the It\={o} generator
matrix}

\label{sec:intro-QSDE}We work in the category of the Hudson and
Parthasarathy (bosonic) quantum stochastic models \cite
{HP84,KRP92,Mey95,GJ07}. Here we fix a separable Hilbert space $\mathfrak{h}$, called
the {\em initial} or {\em system} (Hilbert) space, describing the joint state space of the systems at the nodes of the network,
and a finite-dimensional multiplicity space $\mathfrak{K}$ labelling the input fields. The open
quantum system and the quantum boson fields jointly evolve in a unitary
manner according to the solution of a right Hudson-Parthasarathy quantum
stochastic differential equation (QSDE), using the Einstein summation convention, $$U(t)
=I+\int_{0}^{t}U(s) \,G_{\alpha \beta }dA^{\alpha \beta
}(s), $$ 
with $\alpha,\beta=0,1,2,\ldots,n$ ($n$ denotes the dimension of $\mathfrak{K}$) and $\mathbf{G}=\left[ G_{\alpha \beta }\right] $
is a {\em right} It\={o} generator matrix in the set $\mathfrak{G}\left( \mathfrak{h},\mathfrak{K}%
\right) $ of all right It\={o} generator matrices on systems with initial space $\mathfrak{h}$ and multiplicity space $\mathfrak{K}$; see \cite{GJ08,EGJ} for conventions and notation. Here right (left) QSDE means that the generator $G_{\alpha \beta}dA^{\alpha \beta}(s)$ appears to the right (left) of the unitary $U(t)$. Following \cite{BvHS07}, we work with right unitary processes for technical reasons. The solution $U(t)$ of the QSDEs, when they exist, are adapted quantum
stochastic processes. The right It\={o} generator matrix is written as 
\begin{equation*}
\mathbf{G}=\left[ 
\begin{array}{cc}
K & L \\ 
M & N-I
\end{array}
\right]
\end{equation*}
with respect to the standard decomposition of the coefficient space $\mathfrak{C}=%
\mathfrak{h}\otimes \left( \mathbb{C}\oplus \mathfrak{K}\right)$, that is, as $%
\mathfrak{h}\oplus \left( \mathfrak{h}\otimes \mathfrak{K}\right)$. Here $K=G_{00}$,
$L=[G_{0j}]_{j=1,2,\ldots,n}$, $M=[G_{j0}]_{j=1,2,\ldots,n}$, $N=[G_{jl}]_{j,l=1,2,\ldots,n}$.
Throughout this paper we shall assume that all the components of $K$, $K^{\ast }$, $L$, $%
L^{\ast }$, $M$, $M^{\ast }$, $N$ and $N^{\ast }$ have a common invariant
domain $\mathcal{D}$ in $\mathfrak{h}$ (here $^*$ denotes the adjoint of a Hilbert space operator). We further require that the Hudson-Parthasarathy conditions
are satisfied: $N$ is unitary, $%
K+K^{\ast }=-LL^{\ast }$, and $M=-NL^{\ast }$. Note that if the coefficients
are bounded then these conditions are necessary and sufficient for $U(t)$ to
be a unitary co-cycle (if they are unbounded then the solution may not extend
to a unitary co-cycle). In the general case, if $U(t)$ is a well-defined
unitary and $|\psi _{0}\rangle $ is the initial pure state of the composite
system consisting of the system and the fields at time 0, then this state
vector evolves in time in the Schr\"{o}dinger picture as $|\psi (t)\rangle
=U(t)^{\ast }|\psi _{0}\rangle $. We assume throughout that the operator coefficients
of the QSDE satisfy sufficient conditions that guarantee a unique solution which
extends to a unitary co-cycle on $\mathfrak{h}\otimes \Gamma (L_{\mathfrak{K}%
}^{2}[0,\infty ))$ (in particular this will always be the case when the
coefficients are bounded); see, e.g., \cite{Fagno90,FRS94} for the unbounded case.

Note that $\mathbf{G}$ is simply the adjoint of the corresponding {\em left} It\={o} generator matrices introduced for {\em left} QSDEs in \cite{GJ08}, and plays a similar role to the latter for right QSDEs. Since we will be working exclusively with right QSDEs, from this point on when we say It\={o} generator matrix we will mean the right It\={o} generator matrix. 

We use the notation $X^{-}$ for a generalized inverse of an operator $X\in 
\mathfrak{L}\left( \mathfrak{h}\right) $, that is, $XX^{-}X=X$. Throughout, we require that $X,X^*,X^-,X^{-*}$
have $\mathcal{D}$ as invariant domain. Note then that $X^{-*}=(X^{*})^-$.

\begin{definition}
Given a non-trivial decomposition of the coefficient space $\mathfrak{C}=\mathfrak{C}%
_{1}\oplus \mathfrak{C}_{2}$, we define the generalized Schur complement operation of
It\={o} matrices as 
\begin{equation*}
S_{\mathfrak{C}\mapsto \mathfrak{C}_{1}}\mathbf{G}=G_{11}-G_{12}G_{22}^{-}G_{21}
\end{equation*}
where $\mathbf{G}\equiv \left[ 
\begin{array}{cc}
G_{11} & G_{12} \\ 
G_{21} & G_{22}
\end{array}
\right] $ is the partition of $\mathbf{G}$ with respect to the
decomposition. The domain of $S_{\mathfrak{C}\mapsto \mathfrak{C}_{1}}$ is the set
of $\mathbf{G}\in \mathfrak{L}\left( \mathfrak{C}_{1}\oplus \mathfrak{C}_{2}\right) $
for which we have the image and kernel space inclusions ${\rm im}\left(
G_{21}\right) \subseteq ${\rm im}$\left( G_{22}\right) $ and ${\rm ker}\left(
G_{22}\right) \subseteq  {\rm ker}\left( G_{12}\right) $ (this ensures that the
choice of generalized inverse is unimportant; see \cite{GNW10}  and the references therein). $S_{\mathfrak{C}\mapsto \mathfrak{C}%
_{1}}$ maps into the reduced space $\mathfrak{L}\left( \mathfrak{C}_{1}\right) $. We
shall often use the shorthand $\mathbf{G}/G_{22}$ for the generalized Schur complement.
\end{definition}
Of course, if $G_{22} \mid_{\mathcal{D}}$ is invertible then the generalized Schur complement reduces to the ordinary Schur complement
with the generalized inverse $G_{22}^-$ replaced by $(G_{22} \mid_{\mathcal{D}})^{-1}$.

\section{Eliminating internal connections}

\label{sec:IF-Schur}The total multiplicity space $\mathfrak{K}$ may be
decomposed into external and internal elements as follows 
\begin{equation*}
\mathfrak{K}=\mathfrak{K}_{\mathsf{e}}\oplus \mathfrak{K}_{\mathsf{i}},
\end{equation*}
leading to decomposition $\mathfrak{C}=\mathfrak{C}_{\mathsf{e}}\oplus \mathfrak{C}_{%
\mathsf{i}}$ where $\mathfrak{C}_{\mathsf{e}}=\mathfrak{h}\otimes \left( \mathbb{C}%
\oplus \mathfrak{K}_{\mathsf{e}}\right) $. It was shown in \cite{GJ08} that in
the instantaneous feedback limit for the internal connections, the
reduced It\={o} generator matrix is the Schur complement  of the
pre-interconnection network It\={o} generator matrix, $S_{\mathfrak{C}\mapsto 
\mathfrak{C}_{\mathsf{e}}}\mathbf{G}$. With respect to the decomposition $\mathfrak{C%
}=\mathfrak{h}\oplus \left( \mathfrak{h}\otimes \mathfrak{K}_{\mathsf{e}}\right) \oplus
\left( \mathfrak{h}\otimes \mathfrak{K}_{\mathsf{i}}\right) $, we have, with $L=%
\left[ 
\begin{array}{cc}
L_{\mathsf{e}} & L_{\mathsf{i}}
\end{array}
\right] $, $N_{\mathsf{a}}=\left[ 
\begin{array}{cc}
N_{\mathsf{ae}} & N_{\mathsf{ai}}
\end{array}
\right] $, 
\begin{eqnarray*}
&&\left[ 
\begin{array}{ccc}
K & L_{\mathsf{e}} & L_{\mathsf{i}} \\ 
M_{\mathsf{e}} & N_{\mathsf{ee}}-I & N_{\mathsf{ei}} \\ 
M_{\mathsf{i}} & N_{\mathsf{ie}} & N_{\mathsf{ii}}-I
\end{array}
\right] /\left( N_{\mathsf{ii}}-I\right) \\
&=&\left[ 
\begin{array}{cc}
K & L_{\mathsf{e}} \\ 
M_{\mathsf{e}} & N_{\mathsf{ee}}-I
\end{array}
\right] -\left[ 
\begin{array}{c}
L_{\mathsf{i}} \\ 
N_{\mathsf{ei}}
\end{array}
\right]  \left( N_{\mathsf{ii}}-I\right) ^{-1} \left[ 
\begin{array}{cc}
M_{\mathsf{i}} & N_{\mathsf{ie}}
\end{array}
\right].
\end{eqnarray*}
where it is a condition that $N_{\mathsf{ii}}-I$ be invertible for the
network connections to be well-posed. We denote the operation $S_{\mathfrak{C}%
\mapsto \mathfrak{C}_{\mathsf{e}}}$ of instantaneous feedback reduction by  $%
\mathcal{F}$ whenever the context is clear, and for well-posed connections it
maps between the categories of It\={o} generator matrices in $\mathfrak{G}\left( 
\mathfrak{h},\mathfrak{K}\right) $ to $\mathfrak{G}\left( \mathfrak{h},\mathfrak{K}_{\mathsf{e}%
}\right) $ \cite{GJ08}.

\section{Adiabatic elimination of QSDEs: Structural assumptions}
\label{sec:BvHS}
The following section reviews the adiabatic elimination results of Bouten, van Handel and
Silberfarb \cite{BvHS07}. We consider a QSDE of the form 
$$U^{(k)}(t)
=I+\int_{0}^{t}U^{(k)}(s) G^{(k)}_{\alpha \beta}dA^{\alpha \beta}(s),$$ 
where as before $\alpha,\beta = 0,1,\ldots,n$ and $\mathbf{G}^{(k)}=[G^{(k)}_{\alpha \beta}]$ is an Ito generator matrix $\mathbf{G}^{(k)} \in \mathfrak{G}\left( \mathfrak{h},\mathfrak{K}%
\right)$ that can be expressed as
$$
\mathbf{G}^{(k)}=\left[ \begin{array}{cc} K^{(k)} & L^{(k)}\\M^{(k)} & N^{(k)}-I \end{array}\right]
$$
with $K^{(k)}=G^{(k)}_{00}=k^2Y+kA+B$ and $L^{(k)}=[G^{(k)}_{0j}]_{j=1,2,\ldots,n}=kF+G$, $M^{(k)}=[G^{(k)}_{j 0}]_{j=1,2,\ldots,n}$, and $N^{(k)}=[G^{(k)}_{jl}]_{j,l=1,2,\ldots,n}$, and $k$ is a positive parameter representing coupling strength. The operators $Y$, $A$, $B$, $F$, $G$, $N$, and their respective adjoints, have $\mathcal{D}$ as a common invariant domain, and the coefficients satisfy the Hudson-Parthasarathy conditions $K^{(k)}+K^{(k)*}=-L^{(k)}L^{(k)*}$, $M^{(k)}=-N^{(k)*}L^{(k)}$, and $N^{(k)}N^{(k)*}=N^{(k)*}N^{(k)}=I$. In particular, this implies  that $B+B^{\ast }=-GG^{\ast },A+A^{\ast }=-\left( FG^{\ast
}+GF^{\ast }\right) ,Y+Y^{\ast }=-FF^{\ast }$.  The general situation is that there is a decomposition of the initial/system space $\mathfrak{h}_s$.
into slow and fast subspaces (the subscripts ${\tt s}$ and ${\tt f}$ denote fast and slow, respectively): 
\begin{equation*}
\mathfrak{h}=\mathfrak{h}_{\mathtt{s}}\oplus \mathfrak{h}_{\mathtt{f}},
\end{equation*}
Denote the orthogonal projections onto $\mathfrak{h}_{\mathtt{s}},\mathfrak{h}_{\mathtt{f}}$\
by $P_{\mathtt{s}},P_{\mathtt{f}}$, respectively. With an obvious abuse of
notation, we use the same partition for the decomposition of the coefficient
space: $\mathfrak{C}=\mathfrak{C}_{\mathtt{s}}\oplus \mathfrak{C}_{\mathtt{f}}$ where $%
\mathfrak{C}_{\mathtt{a}}=\mathfrak{h}_{\mathtt{a}}\otimes \left( \mathbb{C}\oplus 
\mathfrak{K}\right) $. With respect to the decomposition $\mathfrak{h}_{\mathtt{s}%
}\oplus \mathfrak{h}_{\mathtt{f}}$, one requires  \cite{BvHS07}:

\begin{enumerate}

\item $P_s \mathcal{D} \subset \mathcal{D}$.

\item  $N^{(k)}=N$ is $k$ independent

\item  $P_{\mathtt{s}}F=0$. That is, $F$ has the structure $
F = \left[ 
\begin{array}{ll}
0 & 0 \\ 
F_{\mathtt{fs}} & F_{\mathtt{ff}}
\end{array}
\right]$.

\item  The Hamiltonian $H^{(k)}=\frac{1}{2i}(K^{(k)}-K^{(k)\ast })$ takes
the form $H\left( k\right) =H^{\left( 0\right) }+kH^{\left( 1\right)
}+k^{2}H^{\left( 2\right) }$ where $P_{\mathtt{s}}H^{\left( 1\right) }P_{%
\mathtt{s}}=0$ and $P_{\mathtt{s}}H^{\left( 2\right) }=H^{(2)}P_{\mathtt{s}%
}=0$, that is, 
\begin{equation*}
H =\left[ 
\begin{array}{ll}
H_{\mathtt{ss}}^{\left( 0\right) }, & H_{\mathtt{sf}}^{\left( 0\right) }+kH_{%
\mathtt{sf}}^{\left( 1\right) } \\ 
H_{\mathtt{fs}}^{\left( 0\right) }+kH_{\mathtt{fs}}^{\left( 1\right) }, & H_{%
\mathtt{ff}}^{\left( 0\right) }+kH_{\mathtt{ff}}^{\left( 1\right) }+k^{2}H_{%
\mathtt{ff}}^{\left( 2\right) }
\end{array}
\right] .
\end{equation*}
Conditions 3 and 4 is equivalent to  $Y$ having the structure 
$Y=\left[\begin{array}{cc} 0 & 0 \\ 0 &  P_{\tt f} Y P_{\tt f} \end{array}\right]$.

\item  In the expansion 
\begin{equation*}
K^{\left( k\right)} =-L^{(k)}\dfrac{1}{2}L^{(k){\ast }}-iH^{(k)} \equiv k^{2}Y+kA+B,
\end{equation*}
we require that the operator $Y_{\mathtt{ff}}=-\dfrac{1}{2}\sum_{\mathtt{a}=%
\mathtt{s},\mathtt{f}}F_{\mathtt{fa}}F_{\mathtt{fa}}^{\ast }-iH_{\mathtt{ff}%
}^{\left( 2\right) }$ is invertible. In particular,  Conditions 3 to 5 is equivalent to $Y$ having a generalized
inverse $Y^-$ with the diagonal structure 
$Y^-=\left[\begin{array}{cc} P_{\tt s} Y^- P_{\tt s} & 0 \\ 0 &  Y_{\tt ff}^{-1} \end{array}\right]$.
\end{enumerate}

Employing a repeated index summation convention over the index range $%
\left\{ \mathtt{s},\mathtt{f}\right\} $ from now on, we find that the
operator $B$ has components $B_{\mathtt{ab}}=-\dfrac{1}{2}G_{\mathtt{ca}}G_{%
\mathtt{cb}}^{\ast }-iH_{\mathtt{ab}}^{\left( 0\right) }$ with respect to
the slow-fast block decomposition. Likewise
\begin{eqnarray}
A &\equiv &\left[ 
\begin{array}{cc}
0 & A_{\mathtt{sf}} \\ 
A_{\mathtt{fs}} & A_{\mathtt{ff}}
\end{array}
\right] =\left[ 
\begin{array}{ll}
0 & -\frac{1}{2}G_{\mathtt{sc}}F_{\mathtt{fc}}^{\ast }-iH_{\mathtt{sf}%
}^{\left( 1\right) } \\ 
-\frac{1}{2}F_{\mathtt{fc}}G_{\mathtt{sc}}^{\ast }-iH_{\mathtt{fs}}^{\left(
1\right) } & -\frac{1}{2}F_{\mathtt{fc}}G_{\mathtt{fc}}^{\ast }-\frac{1}{2}%
G_{\mathtt{fc}}F_{\mathtt{fc}}^{\ast }-iH_{\mathtt{ff}}^{\left( 1\right) }
\end{array}
\right]  \notag \\
Y &\equiv &\left[ 
\begin{array}{ll}
0 & 0 \\ 
0 & Y_{\mathtt{ff}}
\end{array}
\right] .  \label{K relations}
\end{eqnarray}
With respect to
the decomposition the decomposition $\mathfrak{C}=\mathfrak{h}_{\mathtt{s}}\oplus
\left( \mathfrak{h}_{\mathtt{s}}\otimes \mathfrak{K}\right) \oplus \mathfrak{h}_{\mathtt{%
f}}\oplus \left( \mathfrak{h}_{\mathtt{f}}\otimes \mathfrak{K}\right) $ we have 
\begin{equation}
\mathbf{G}^{\left( k\right) }=[ 
\begin{array}{cccc}
1 & 1 & k & 1
\end{array}
]\left[ \mathbf{G}_{0}+\mathbf{G}^{\prime }\left( k\right) \right] \left[ 
\begin{array}{c}
1 \\ 
1 \\ 
k \\ 
1
\end{array}
\right]  \label{eq:G(k) partition  slow fast}
\end{equation}
where 
\begin{equation*}
\mathbf{G}_{0}=\left[ 
\begin{array}{cccc}
B_{\mathtt{ss}} & G_{\mathtt{ss}} & A_{\mathtt{sf}} & G_{\mathtt{sf}} \\ 
-N_{\mathtt{sa}}G_{\mathtt{sa}}^{\ast } & N_{\mathtt{ss}}-I & -N_{\mathtt{sa}%
}F_{\mathtt{fa}}^{\ast } & N_{\mathtt{sf}} \\ 
A_{\mathtt{fs}} & F_{\mathtt{fs}} & Y_{\mathtt{ff}} & F_{\mathtt{ff}} \\ 
-N_{\mathtt{fa}}G_{\mathtt{sa}}^{\ast } & N_{\mathtt{fs}} & -N_{\mathtt{fa}%
}F_{f\mathtt{a}}^{\ast } & N_{\mathtt{ff}}-I
\end{array}
\right] ,
\end{equation*}
and $\lim_{k\rightarrow \infty }\mathbf{G}^{\prime }\left( k\right) \phi =0$
for all $\phi \in \mathcal{D}$. We then observe that 
\begin{equation*}
\mathbf{G}_{0}/Y_{\mathtt{ff}}=\left[ 
\begin{array}{ccc}
\hat{K}_{\mathtt{ss}} & \hat{L}_{\mathtt{s}} & \hat{L}_{\mathtt{f}} \\ 
\hat{M}_{\mathtt{s}} & \hat{N}_{\mathtt{ss}}-I & \hat{N}_{\mathtt{sf}} \\ 
\hat{M}_{\mathtt{f}} & \hat{N}_{\mathtt{fs}} & \hat{N}_{\mathtt{ff}}-I
\end{array}
\right]
\end{equation*}
where 
\begin{eqnarray*}
\hat{K}_{\mathtt{ss}} &=&B_{\mathtt{ss}}-A_{\mathtt{sf}}Y_{\mathtt{ff}%
}^{-1}A_{\mathtt{fs}},\;\hat{L}_{\mathtt{a}}=G_{\mathtt{sa}}-A_{\mathtt{sf}%
}Y_{\mathtt{ff}}^{-1}F_{\mathtt{fa}}, \\
\hat{M}_{\mathtt{a}} &=&-N_{\mathtt{ab}}G_{\mathtt{sb}}^{\dag }+N_{\mathtt{ab%
}}F_{\mathtt{fb}}^{\ast }Y_{\mathtt{ff}}^{-1}A_{\mathtt{fs}},\;\hat{N}_{%
\mathtt{ab}}=N_{\mathtt{ab}}+N_{\mathtt{ac}}F_{\mathtt{fc}}^{\ast }Y_{%
\mathtt{ff}}^{-1}F_{\mathtt{fb}}.
\end{eqnarray*}
We also assume that 
\begin{equation}
\hat{L}_{\mathtt{f}}=\hat{N}_{\mathtt{sf}}=\hat{N}_{\mathtt{fs}}=0,
\label{eq:fast}
\end{equation}
and this will ensure that the limit dynamics excludes the possibility of
transitions that terminate in any of the fast states. In this case $\hat{N}_{%
\mathtt{ss}}$ and $\hat{N}_{\mathtt{ff}}$ are unitary. In particular 
\begin{equation}
\mathbf{\hat{G}}=\left[ 
\begin{array}{cc}
\hat{K}_{\mathtt{ss}} & \hat{L}_{\mathtt{s}} \\ 
\hat{M}_{\mathtt{s}} & \hat{N}_{\mathtt{ss}}-I
\end{array}
\right] \equiv \left[ 
\begin{array}{cc}
\hat{K} & \hat{L} \\ 
\hat{M} & \hat{N}-I
\end{array}
\right]  \label{eg:hats}
\end{equation}
is an It\={o} generator matrix $\left( \hat{M}_{\mathtt{s}}=-\hat{N}_{%
\mathtt{ss}}\hat{L}_{\mathtt{s}}^{\ast }\right) $ on the coefficient space $%
\mathfrak{C}_{\mathtt{s}}=\mathfrak{h}_{\mathtt{s}}\otimes \left( \mathbb{C}\oplus 
\mathfrak{K}\right) $. The final assumption is a technical condition. 
For any $%
\alpha ,\beta \in \mathbb{C}^{n}$ (represented as column vectors), $P_{%
\mathtt{s}}\mathcal{D}$ is a core for the operator $\mathcal{L}^{(\alpha
\beta )}$ defined by: 
\begin{equation}
\mathcal{L}^{(\alpha \beta )}=\alpha ^{\ast }\hat{N}\beta +\alpha ^{\ast }%
\hat{M}+\hat{L}\beta +\hat{K}-\frac{|\alpha |^{2}+|\beta |^{2}}{2}, \label{cond:core} 
\end{equation}
with $\hat{K},\hat{L},\hat{M},\hat{N}$ as defined in (\ref{eg:hats}).

\begin{theorem}[\protect\cite{BvHS07}]
\label{thm:s-conv} Suppose that all the assumptions above hold. If the right
QSDEs with coefficients $\mathbf{G}^{\left( k\right) }$ possess a unique
solution that extends to a contraction co-cycle $U^{(k)}(t)$ on $\mathfrak{h}%
\otimes \Gamma \left( L_{\mathfrak{K}}^{2}[0,\infty )\right) $ for all $k>0$,
and the right QSDE with coefficients $\mathbf{\hat{G}}$ has a unique
solution that extends to a unitary co-cycle $\hat{U}(t)$ on $\mathfrak{h}_{%
\mathtt{s}}\otimes \Gamma \left( L_{\mathfrak{K}}^{2}[0,\infty )\right) $, then $%
U^{(k)}(t)$ converges to the solution $\hat{U}(t)$ uniformly in a strong
sense: 
\begin{equation*}
\mathop{\lim}_{k\rightarrow \infty }\mathop{\sup}_{0\leq t\leq T}\Vert
U^{(k)}(t)^{\ast }\phi -\hat{U}(t)^{\ast }\phi \Vert =0,\quad \forall \phi
\in \mathfrak{h}_{\mathtt{s}}\otimes \Gamma (L_{\mathfrak{K}%
}^{2}[0,\infty )),
\end{equation*}
for each fixed $T\geq 0$.
\end{theorem}

The above theorem is Theorem 3 of \cite{BvHS07}.

\section{Adiabatic elimination of QSDEs: Schur complements}
\label{subsec:AE-Schur} In this section we will show how the singular
perturbation limit of the QSDE can be related to the Schur complementation
of a certain matrix with operator entries. To this end, define the extended
It\={o} generator matrix $\mathbf{G}_{E}$ as: 
\begin{equation*}
\mathbf{G}_{E}=\left[ 
\begin{array}{ccc}
B & A_{\mathtt{sf}} & G \\ 
A_{\mathtt{f}} & Y_{\mathtt{ff}} & F_{\mathtt{f}} \\ 
-NG^{\ast } & -NF_{\mathtt{f}}^{\ast } & N-I
\end{array}
\right] ,
\end{equation*}
where $A_{\mathtt{f}}=P_{\mathtt{f}}A$, $F_{\mathtt{f}}=P_{\mathtt{f}}F$.

\begin{lemma}
\label{lm:AE-Schur} The limit QSDE $\hat{U}(t)$ has the It\={o} generator
matrix $\hat{\mathbf{G}}$ given by $\hat{\mathbf{G}}=P_{\mathtt{s}}(\mathbf{G%
}_{E}/Y_{\mathtt{ff}})P_{\mathtt{s}}\mid _{\mathfrak{h}_{\mathtt{s}}}$, where $%
\mathbf{G}_{E}/Y_{\mathtt{ff}}$ is the Schur complement of $\mathbf{G}_{E}$
with respect to the sub-block with entry $Y_{\mathtt{ff}}$.
\end{lemma}

\begin{proof}
Direct calculation shows that 
\begin{align}
\mathbf{G}_{E}/Y_{\mathtt{ff}}& =\left[ 
\begin{array}{cc}
B & G \\ 
-NG^{\ast } & N-I
\end{array}
\right] -\left[ 
\begin{array}{c}
A_{\mathtt{sf}} \\ 
-NF_{\mathtt{f}}^{\ast }
\end{array}
\right] Y_{\mathtt{ff}}^{-1}\left[ 
\begin{array}{cc}
A_{\mathtt{f}} & F
\end{array}
\right] ,  \notag \\
& =\left[ 
\begin{array}{cc}
B-A_{\mathtt{sf}}Y_{\mathtt{ff}}^{-1}A_{\mathtt{f}} & G-A_{\mathtt{sf}}Y_{%
\mathtt{ff}}^{-1}F_{\mathtt{f}} \\ 
-NG^{\ast }+NF_{\mathtt{f}}^{\ast }Y_{\mathtt{ff}}^{-1}A_{\mathtt{f}} & 
N+NF_{\mathtt{f}}^{\ast }Y_{\mathtt{ff}}^{-1}F_{\mathtt{f}}-I
\end{array}
\right] .  \label{eq:SC-GE-by-Y}
\end{align}
Thus: 
\begin{equation*}
P_{\mathtt{s}}(\mathbf{G}_{E}/Y_{\mathtt{ff}})P_{\mathtt{s}}=\left[ 
\begin{array}{cc}
P_{\mathtt{s}}(B-A_{\mathtt{sf}}Y_{\mathtt{ff}}^{-1}A_{\mathtt{f}})P_{%
\mathtt{s}} & P_{\mathtt{s}}(G-A_{\mathtt{sf}}Y_{\mathtt{ff}}^{-1}F_{\mathtt{%
f}})P_{\mathtt{s}} \\ 
P_{\mathtt{s}}(-NG^{\ast }+NF_{\mathtt{f}}^{\ast }Y_{\mathtt{ff}}^{-1}A_{%
\mathtt{f}})P_{\mathtt{s}} & P_{\mathtt{s}}(N+NF_{\mathtt{f}}^{\ast }Y_{%
\mathtt{ff}}^{-1}F_{\mathtt{f}})P_{\mathtt{s}}-P_{\mathtt{s}}
\end{array}
\right] .
\end{equation*}
Therefore, since $P_{\mathtt{s}}(\mathbf{G}_{E}/Y_{\mathtt{ff}})P_{\mathtt{s}}\mid _{\mathfrak{h}_{%
\mathtt{s}}}$ equals
\begin{equation}
\left[ 
\begin{array}{cc}
P_{\mathtt{s}}(B-A_{\mathtt{sf}}Y_{\mathtt{ff}}^{-1}A_{\mathtt{f}})P_{%
\mathtt{s}} & P_{\mathtt{s}}(G-A_{\mathtt{sf}}Y_{\mathtt{ff}}^{-1}F_{\mathtt{%
f}})P_{\mathtt{s}} \\ 
P_{\mathtt{s}}(-NG^{\ast }+NF_{\mathtt{f}}^{\ast }Y_{\mathtt{ff}}^{-1}A_{%
\mathtt{f}})P_{\mathtt{s}} & P_{\mathtt{s}}(N+NF_{\mathtt{f}}^{\ast }Y_{%
\mathrm{ff}}^{-1}F_{\mathtt{f}})P_{\mathtt{s}}-I
\end{array}
\right] ,  \label{eq:lim-l-GM-AE}
\end{equation}
it follows from (\ref{eq:fast}) that $\hat{\mathbf{G}}=P_{\mathtt{s}}(%
\mathbf{G}_{E}/Y_{\mathtt{ff}})P_{\mathtt{s}}\mid _{\mathfrak{h}_{\mathtt{s}}}$
\end{proof}

We then we denote by $\mathcal{A}$ the map that takes $\mathbf{G}^{(k)}$ to
the It\={o} generator matrix $\hat{\mathbf{G}}$ in the lemma by: $\mathcal{A}:\mathbf{G}%
^{(k)}\mapsto \hat{\mathbf{G}}$.

We conclude by remarking that the instantaneous feedback limit operation $\mathcal{F}$
and the adiabatic elimination operations  $\mathcal{A}$ can be cast as structure preserving transformations,
that is, transformations that preserve the structure of It\={o} generators matrices or convert It\={o} generator
matrices to It\={o} generator matrices (possibly of  lower initial space and multiplicity space dimensions). 
 
\section{Sequential application of the instantaneous feedback and adiabatic
elimination operations}

\label{sec:IF-AE-sim}

\subsection{The adiabatic elimination operation followed by the
instantaneous feedback operation}

When the adiabatic elimination operation is first applied followed by the
instantaneous feedback operation we have the following:

\begin{lemma}
\label{lm:IF-AE} Under the standing assumptions in Section \ref{sec:BvHS}, and taking $N_{\mathsf{ii}}+N_{\mathsf{i}}F_{%
\mathrm{f}}^{\ast }Y_{\mathtt{ff}}^{-1}F_{\mathtt{f}\mathsf{i}}-I$ to be
invertible, we have 
\begin{equation*}
P_{\mathtt{s}}\bigl((\mathbf{G}_{E}/Y_{\mathtt{ff}})/(N_{\mathsf{ii}}+N_{%
\mathsf{i}}F_{\mathtt{f}}^{\ast }Y_{\mathtt{ff}}^{-1}F_{\mathtt{f}\mathsf{i}%
}-I)\bigr)P_{\mathtt{s}}=\mathcal{F}\mathcal{A}\mathbf{G}^{(k)},
\end{equation*}
where $F_{\mathtt{f}\mathsf{i}}=P_{\mathtt{f}}F_{\mathsf{i}}$.
\end{lemma}

\begin{proof}
Partition the extended It\={o} generator with respect to $\mathfrak{K}_{\mathsf{e%
}}\oplus \mathfrak{K}_{\mathsf{i}}$ to get 
\begin{gather*}
\mathbf{G}_{E}/Y_{\mathtt{ff}}=\left[ 
\begin{array}{cccc}
B & A_{\mathtt{sf}} & G_{\mathsf{i}} & G_{\mathsf{e}} \\ 
A_{\mathtt{f}} & Y_{\mathtt{ff}} & F_{\mathtt{f}\mathsf{i}} & F_{\mathtt{f}%
\mathsf{e}} \\ 
-N_{\mathsf{i}}G^{\ast } & -N_{\mathsf{i}}F_{\mathtt{f}}^{\ast } & N_{%
\mathsf{ii}}-I & N_{\mathsf{ie}} \\ 
-N_{\mathsf{e}}G^{\ast } & -N_{\mathsf{e}}F_{\mathtt{f}}^{\ast } & N_{%
\mathsf{ei}} & N_{\mathsf{ee}}-I
\end{array}
\right] /Y_{\mathtt{ff}} \\
=\left[ 
\begin{array}{cc}
B-A_{\mathtt{sf}}Y_{\mathtt{ff}}^{-1}A_{\mathtt{f}} & G_{\mathsf{i}}-A_{%
\mathtt{sf}}Y_{\mathtt{ff}}^{-1}F_{\mathtt{f}\mathsf{i}} \\ 
-N_{\mathsf{i}}G^{\ast }+N_{\mathsf{i}}F_{\mathtt{f}}^{\ast }Y_{\mathtt{ff}%
}^{-1}A_{\mathtt{f}} & N_{\mathsf{ii}}+N_{\mathsf{i}}F_{\mathtt{f}}^{\ast
}Y_{\mathtt{ff}}^{-1}F_{\mathtt{f}\mathsf{i}}-I \\ 
-N_{\mathsf{e}}G^{\ast }+N_{\mathsf{e}}F_{\mathtt{f}}^{\ast }Y_{\mathtt{ff}%
}^{-1}A_{\mathtt{f}} & N_{\mathsf{ei}}+N_{\mathsf{e}}F_{\mathtt{f}}^{\ast
}Y_{\mathtt{ff}}^{-1}F_{\mathtt{f}\mathsf{i}}
\end{array}
\right.  \\
\left. 
\begin{array}{c}
G_{\mathsf{e}}-A_{\mathtt{sf}}Y_{\mathtt{ff}}^{-1}F_{\mathtt{f}\mathsf{e}}
\\ 
N_{\mathsf{ie}}+N_{\mathsf{i}}F_{\mathtt{f}}^{\ast }Y_{\mathtt{ff}}^{-1}F_{%
\mathtt{f}\mathsf{e}} \\ 
N_{\mathsf{ee}}+N_{\mathsf{e}}F_{\mathtt{f}}^{\ast }Y_{\mathtt{ff}}^{-1}F_{%
\mathtt{f}\mathsf{e}}-I
\end{array}
\right] \equiv \left[ 
\begin{array}{ccc}
\hat{B} & \hat{G}_{\mathsf{i}} & \hat{G}_{\mathsf{e}} \\ 
\hat{M}_{\mathsf{i}} & \hat{N}_{\mathsf{ii}}-I & \hat{N}_{\mathsf{ie}} \\ 
\hat{M}_{\mathsf{e}} & \hat{N}_{\mathsf{ei}} & \hat{N}_{\mathsf{ee}}-I
\end{array}
\right] 
\end{gather*}
where $N_{\mathsf{a}}=\left[\begin{array}{cc}  N_{\mathsf{ae}} & N_{\mathsf{ai}}\end{array} \right] $, $F_{%
\mathtt{f}\mathsf{a}}=P_{\mathtt{f}}F_{\mathsf{a}}$ for $\mathsf{a}=\mathsf{i%
},\mathsf{e}$, and $[
\begin{array}{cc}
F_{\mathsf{i}} & F_{\mathsf{e}}
\end{array}
]=F$ and $[
\begin{array}{cc}
G_{\mathsf{i}} & G_{\mathsf{e}}
\end{array}
]=G$, and we used (\ref{eq:SC-GE-by-Y}). We now apply the operation $%
\mathcal{F}$ to get $(\mathbf{G}_{E}/Y_{\mathtt{ff}})$
$/\left( \hat{N}_{\mathsf{ii} }-I\right)$ equal to
\begin{equation*}
 \left[ 
\begin{array}{cc}
\hat{B}-\hat{G}_{\mathsf{i}}\left( \hat{N}_{\mathsf{ii}}-I\right) ^{-1}\hat{M%
}_{\mathsf{i}} & \hat{G}_{\mathsf{e}}-\hat{G}_{\mathsf{i}}\left( \hat{N}_{%
\mathsf{ii}}-I\right) ^{-1}\hat{N}_{\mathsf{ie}} \\ 
\hat{M}_{\mathsf{e}}-\hat{N}_{\mathsf{ei}}\left( \hat{N}_{\mathsf{ii}%
}-I\right) ^{-1}\hat{M}_{\mathsf{i}} & \hat{N}_{\mathsf{ee}}-\hat{N}_{%
\mathsf{ei}}\left( \hat{N}_{\mathsf{ii}}-I\right) ^{-1}\hat{N}_{\mathsf{ie}}-I
\end{array}
\right] .
\end{equation*}

Next, note that $N_{\mathsf{ii}}+N_{\mathsf{i}}F_{\mathtt{f}}^{\ast }Y_{%
\mathtt{ff}}^{-1}F_{\mathtt{f}\mathsf{i}}$ has the representation 
\begin{equation*}
N_{\mathsf{ii}}+N_{\mathsf{i}}F_{\mathtt{f}}^{\ast }Y_{\mathtt{ff}}^{-1}F_{%
\mathtt{f}\mathsf{i}}=\left[ 
\begin{array}{cc}
P_{\mathtt{f}}(N_{\mathsf{ii}}+N_{\mathsf{i}}F_{\mathtt{f}}^{\ast }Y_{%
\mathtt{ff}}^{-1}F_{\mathtt{f}\mathsf{i}})P_{\mathtt{f}} & 0 \\ 
0 & P_{\mathtt{s}}(N_{\mathsf{ii}}+N_{\mathsf{i}}F_{\mathtt{f}}^{\ast }Y_{%
\mathtt{ff}}^{-1}F_{\mathtt{f}\mathsf{i}})P_{\mathtt{s}}
\end{array}
\right] ,
\end{equation*}
with respect to the decomposition $\mathcal{D}=P_{\mathtt{f}}\mathcal{D}%
\oplus P_{\mathtt{s}}\mathcal{D}$. Moreover, we also note the representation 
\begin{equation*}
G_{\mathsf{a}}-A_{\mathtt{sf}}Y_{\mathtt{ff}}^{-1}F_{\mathtt{f}\mathsf{a}}=%
\left[ 
\begin{array}{cc}
P_{\mathtt{f}}G_{\mathsf{a}}P_{\mathtt{f}} & P_{\mathtt{f}}G_{\mathsf{a}}P_{%
\mathtt{s}} \\ 
0 & P_{\mathtt{s}}(G_{\mathsf{a}}-A_{\mathtt{sf}}Y_{\mathtt{ff}}^{-1}F_{%
\mathtt{f}\mathsf{a}})P_{\mathtt{s}}
\end{array}
\right] ,\;\mathsf{a}=\mathsf{i},\mathsf{e}.
\end{equation*}
Using these representations we can verify the following sequence of
identities: 
\begin{multline*}
P_{\mathtt{s}}(\mathbf{G}_{E}/Y_{\mathtt{ff}})/\left( \hat{N}_{\mathsf{ii}%
}-I\right) P_{\mathtt{s}}=P_{\mathtt{s}}(\mathbf{G}_{E}/Y_{\mathtt{ff}})P_{%
\mathtt{s}}/P_{\mathtt{s}}(N_{\mathsf{ii}}+N_{\mathsf{i}}F_{\mathtt{f}%
}^{\ast }Y_{\mathtt{ff}}^{-1}F_{\mathtt{f}\mathsf{i}}-I)P_{\mathtt{s}}, \\
=(\left. P_{\mathtt{s}}(\mathbf{G}_{E}/Y_{\mathtt{ff}})P_{\mathtt{s}}\right|
_{\mathfrak{h}_{\mathtt{s}}})/(P_{\mathtt{s}}(N_{\mathsf{ii}}+N_{\mathsf{i}%
}F_{\mathtt{f}}^{\ast }Y_{\mathtt{ff}}^{-1}F_{\mathtt{f}\mathsf{i}})P_{%
\mathtt{s}}-I),
\end{multline*}
where the last equality follows from the fact that $P_{\mathtt{s}}(N_{%
\mathsf{ii}}+N_{\mathsf{i}}F_{\mathtt{f}}^{\ast }Y_{\mathtt{ff}}^{-1}F_{%
\mathtt{f}\mathsf{i}}-I)P_{\mathtt{s}}\mid _{\mathfrak{h}_{\mathtt{s}}}=P_{%
\mathtt{s}}(N_{\mathsf{ii}}+N_{\mathsf{i}}F_{\mathtt{f}}^{\ast }Y_{\mathtt{ff%
}}^{-1}F_{\mathtt{f}\mathsf{i}})P_{\mathtt{s}}-I$. Finally, since 
\begin{equation*}
\mathcal{F}\mathcal{A}\mathbf{G}^{(k)}=(P_{\mathtt{s}}(\mathbf{G}_{E}/Y_{%
\mathtt{ff}})P_{\mathtt{s}}\mid _{\mathfrak{h}_{\mathtt{s}}})/(P_{\mathtt{s}}(N_{%
\mathsf{ii}}+N_{\mathsf{i}}F_{\mathtt{f}}^{\ast }Y_{\mathtt{ff}}^{-1}F_{%
\mathtt{f}\mathsf{i}})P_{\mathtt{s}}-I),
\end{equation*}
by definition, we thus obtain the desired result.
\end{proof}

\subsection{The instantaneous feedback operation followed by the adiabatic
elimination operation}

We now turn to consider the alternative sequence of first applying the
instantaneous feedback operation followed by the adiabatic elimination
operation. The main result in this section is the following:

\begin{lemma}
\label{lm:AE-IF} Suppose that the assumptions of Section \ref{sec:BvHS} are satisfied, $N_{\mathsf{ii}%
}-I$ is invertible, $\mathsf{\ker }(Y+F_{\mathsf{i}}(N_{\mathsf{ii}%
}-I)^{-1}F_{\mathsf{i}})=\mathfrak{h}_{\mathtt{s}}$, and there exists an
operator $\hat{Y}^{-}$ such that $\hat{Y}^{-},\hat{Y}^{-\ast }$ have $%
\mathcal{D}$ as a common invariant domain and $\hat{Y}\hat{Y}^{-}=\hat{Y}^{-}%
\hat{Y}=P_{\mathtt{f}}$, where $\hat{Y}=Y+F_{\mathsf{i}}(N_{\mathsf{ii}%
}-I)^{-1}F_{\mathsf{i}}$. Then 
\begin{equation*}
\mathcal{A}\mathcal{F}\mathbf{G}^{(k)}=P_{\mathtt{s}}((\mathbf{G}_{E}/(N_{%
\mathsf{ii}}-I))/(Y_{\mathtt{ff}}+F_{\mathtt{f}\mathsf{i}}(N_{\mathsf{ii}%
}-I)^{-1}N_{\mathsf{i}}F_{\mathtt{f}}^{\ast }))P_{\mathtt{s}}\mid _{\mathfrak{h}%
_{\mathtt{s}}}.
\end{equation*}
\end{lemma}
\begin{proof}
We first compute the extended It\={o} generator matrix corresponding to $\mathcal{F}%
\mathbf{G}^{(k)}$. With $\tilde{N}_{\mathsf{ee}}=N_{\mathsf{ee}}-N_{\mathsf{%
ei}}(N_{\mathsf{ii}}-I)^{-1}N_{\mathsf{ie}}$ this is 
\begin{align*}
\lefteqn{(\mathcal{F}\mathbf{G}^{(k)})_{E}} \\
& =\left[ 
\begin{array}{c}
B+G_{\mathsf{i}}(N_{\mathsf{ii}}-I)^{-1}N_{\mathsf{i}}G^{\ast } \\ 
A_{\mathtt{f}}+F_{\mathtt{f}\mathsf{i}}(N_{\mathsf{ii}}-I)^{-1}N_{\mathsf{i}%
}G^{\ast }+P_{\mathtt{f}}G_{\mathsf{i}}(N_{\mathsf{ii}}-I)^{-1}N_{\mathsf{i}%
}F^{\ast } \\ 
-\tilde{N}_{\mathsf{ee}}(G_{\mathsf{e}}-G_{\mathsf{i}}(N_{\mathsf{ii}%
}-I)^{-1}N_{\mathsf{ie}})^{\ast }
\end{array}
\right. \\
& \quad \left. 
\begin{array}{c}
A_{\mathtt{sf}}+P_{\mathtt{s}}G_{\mathsf{i}}(N_{\mathsf{ii}}-I)^{-1}N_{%
\mathsf{i}}F_{\mathtt{f}}^{\ast } \\ 
Y_{\mathtt{ff}}+F_{\mathtt{f}\mathsf{i}}(N_{\mathsf{ii}}-I)^{-1}N_{\mathsf{i}%
}F_{\mathtt{f}}^{\ast } \\ 
-\tilde{N}_{\mathsf{ee}}(F_{\mathtt{f}\mathsf{e}}-F_{\mathtt{f}\mathsf{i}%
}(N_{\mathsf{ii}}-I)^{-1}N_{\mathsf{ie}})^{\ast }
\end{array}
\begin{array}{c}
G_{\mathsf{e}}-G_{\mathsf{i}}(N_{\mathsf{ii}}-I)^{-1}N_{\mathsf{ie}} \\ 
F_{\mathtt{f}\mathsf{e}}-F_{\mathtt{f}\mathsf{i}}(N_{\mathsf{ii}}-I)^{-1}N_{%
\mathsf{ie}} \\ 
\tilde{N}_{\mathsf{ee}}-I
\end{array}
\right]
\end{align*}
Let $\hat{Y}=Y+F_{\mathrm{i}}(N_{\mathrm{ii}}-I)^{-1}N_{\mathrm{i}}F^{\ast }$. 
Then under the structural assumptions of Section \ref{sec:BvHS} and the hypothesis that \textrm{ker}$
(Y+F_{\mathsf{i}}(N_{\mathsf{ii}}-I)^{-1}N_{\mathsf{i}}F^{\ast })\,=\,$%
\textrm{ker}$(Y)$, we have that $\hat{Y}$ has a representation, with respect to the decomposition $\mathcal{D}=P_{\mathtt{f}}\mathcal{D}%
\oplus P_{\mathtt{s}}\mathcal{D}$, with the special structure: 
\begin{equation*}
\hat{Y}=\left[ 
\begin{array}{cc}
P_{\mathtt{f}}YP_{\mathtt{f}}+P_{\mathtt{f}}F_{\mathsf{i}}(N_{\mathsf{ii}%
}-I)^{-1}N_{\mathsf{i}}F^{\ast }P_{\mathtt{f}} & 0 \\ 
0 & 0
\end{array}
\right] =\left[ 
\begin{array}{cc}
Y_{\mathtt{ff}}+F_{\mathtt{f}\mathsf{i}}(N_{\mathsf{ii}}-I)^{-1}N_{\mathsf{i}}F_{\mathtt{f}%
}^{\ast } & 0 \\ 
0 & 0
\end{array}
\right] .
\end{equation*}
Moreover, since there exists an operator $\hat{Y}^{-}$ that satisfy the
hypothesis of the theorem we have that $\hat{Y}^{-}=(Y+F_{\mathsf{i}}(N_{%
\mathsf{ii}}-I)^{-1}N_{\mathsf{i}}F^{\ast })^{-}$ with respect to the same decomposition has the diagonal structure 
\begin{equation*}
\hat{Y}^{-}=\left[ 
\begin{array}{cc}
P_{\mathtt{f}}\hat{Y}^{-}P_{\mathtt{f}} & 0 \\ 
0 & P_{\mathtt{s}}\hat{Y}^{-}P_{\mathtt{s}}
\end{array}
\right] ,
\end{equation*}
with $\hat{Y}_{\mathtt{ff}}=P_{\mathtt{f}}\hat{Y}^{-}P_{\mathtt{f}}$
invertible. In fact, we have that 
\begin{equation*}
\hat{Y}_{\mathtt{ff}}=(Y_{\mathtt{ff}}+F_{\mathtt{f}\mathsf{i}}(N_{\mathsf{ii%
}}-I)^{-1}N_{\mathsf{i}}F_{\mathtt{f}}^{\ast })^{-1}.
\end{equation*}
Introduce the additional notations 
\begin{align*}
\hat{A}_{\mathtt{sf}}& =A_{\mathtt{sf}}+P_{\mathtt{s}}G_{\mathsf{i}}(N_{%
\mathsf{ii}}-I)^{-1}N_{\mathsf{i}}F_{\mathtt{f}}^{\ast }, \\
\hat{A}_{\mathtt{f}}& =A_{\mathtt{f}}+F_{\mathtt{f}\mathsf{i}}(N_{\mathsf{ii}%
}-I)^{-1}N_{\mathsf{i}}G^{\ast }, \\
\hat{F}_{\mathtt{f}}& =F_{\mathtt{f}\mathsf{e}}-F_{\mathtt{f}\mathsf{i}}(N_{%
\mathsf{ii}}-I)^{-1}N_{\mathsf{ie}}.
\end{align*}
From the partitioning of $(\mathbf{G}^{(k)}/(N_{\mathsf{ii}}-I))_{E}$ we can
compute $\mathcal{A}\mathcal{F}\mathbf{G}^{(k)}$ by Lemma \ref{lm:AE-Schur}
as 
\begin{equation*}
\mathcal{A}\mathcal{F}\mathbf{G}^{(k)}=P_{\mathtt{s}}\bigl((\mathbf{G}%
^{(k)}/(N_{\mathsf{ii}}-I))_{E}/\hat{Y}_{\mathtt{ff}}\bigr)\mid _{%
\mathfrak{h}_{s}}\equiv \left[ 
\begin{array}{cc}
\tilde{K} & \tilde{L} \\ 
\tilde{M} & \tilde{N}-I
\end{array}
\right] ,
\end{equation*}
where 
\begin{eqnarray*}
\tilde{K} &=&P_{\mathtt{s}}(B+G_{\mathsf{i}}(N_{\mathsf{ii}}-I)^{-1}N_{%
\mathsf{i}}G^{\ast })P_{\mathtt{s}}-P_{\mathtt{s}}\hat{A}_{\mathtt{sf}}\hat{Y%
}_{\mathtt{ff}}^{-1}\hat{A}_{\mathtt{f}}P_{\mathtt{s}} \\
\tilde{L} &=&P_{\mathtt{s}}(G_{\mathsf{e}}-G_{\mathsf{i}}(N_{\mathsf{ii}%
}-I)^{-1}N_{\mathsf{ie}})P_{\mathtt{s}}-P_{\mathtt{s}}\hat{A}_{\mathtt{sf}}%
\hat{Y}_{\mathtt{ff}}^{-1}\hat{F}_{\mathtt{f}}P_{\mathtt{s}} \\
\tilde{M} &=&-P_{\mathtt{s}}\tilde{N}_{\mathsf{ee}}(G_{\mathsf{e}}-G_{%
\mathsf{i}}(N_{\mathsf{ii}}-I)^{-1}N_{\mathsf{ie}})^{\ast }P_{\mathtt{s}}+P_{%
\mathtt{s}}\tilde{N}_{\mathsf{ee}}\hat{F}_{\mathtt{f}}^{\ast }\hat{Y}_{%
\mathtt{ff}}^{-1}\hat{A}_{\mathtt{f}}P_{\mathtt{s}} \\
\tilde{N} &=&P_{\mathtt{s}}\tilde{N}_{\mathsf{ee}}P_{\mathtt{s}}+P_{\mathtt{s%
}}\tilde{N}_{\mathsf{ee}}\hat{F}_{\mathtt{f}}^{\ast }\hat{Y}_{\mathtt{ff}%
}^{-1}\hat{F}_{\mathtt{f}}P_{\mathtt{s}}
\end{eqnarray*}

We also compute $(\mathbf{G}_{E}/(N_{\mathsf{ii}}-I))/(Y_{\mathtt{ff}}+F_{%
\mathtt{f}\mathsf{i}}(N_{\mathsf{ii}}-I)^{-1}N_{\mathsf{i}}F_{\mathtt{f}%
}^{\ast })$. To begin with $\mathbf{G}_{E}/(N_{\mathsf{ii}}-I)$ is given by 
\begin{eqnarray*}
&&\left[ 
\begin{array}{cc}
B+G_{\mathsf{i}}(N_{\mathsf{ii}}-I)^{-1}N_{\mathsf{i}}G^{\ast } & A_{\mathtt{%
sf}}+G_{\mathsf{i}}(N_{\mathsf{ii}}-I)^{-1}N_{\mathsf{i}}F_{\mathtt{f}%
}^{\ast } \\ 
A_{\mathtt{f}}+F_{\mathtt{f}\mathsf{i}}(N_{\mathsf{ii}}-I)^{-1}N_{\mathsf{i}%
}G^{\ast } & Y_{\mathtt{ff}}+F_{\mathtt{f}\mathsf{i}}(N_{\mathsf{ii}%
}-I)^{-1}N_{\mathsf{i}}F_{\mathtt{f}}^{\ast } \\ 
-\tilde{N}_{\mathsf{ee}}(G_{\mathsf{e}}-G_{\mathsf{i}}(N_{\mathsf{ii}%
}-I)^{-1}N_{\mathsf{ie}})^{\ast } & -\tilde{N}_{\mathsf{ee}}(F_{\mathtt{f}%
\mathsf{e}}-F_{\mathtt{f}\mathsf{i}}(N_{\mathsf{ii}}-I)^{-1}N_{\mathsf{ie}%
})^{\ast }
\end{array}
\right.  \\
&&\left. 
\begin{array}{c}
G_{\mathsf{e}}-G_{\mathsf{i}}(N_{\mathsf{ii}}-I)^{-1}N_{\mathsf{ie}} \\ 
F_{\mathtt{f}\mathsf{e}}-F_{\mathtt{f}\mathsf{i}}(N_{\mathsf{ii}}-I)^{-1}N_{%
\mathsf{ie}} \\ 
\tilde{N}_{\mathsf{ee}}-I
\end{array}
\right] ,
\end{eqnarray*}
Continuing the calculation we then find that 
\begin{gather*}
(\mathbf{G}_{E}/(N_{\rm ii}-I))/(Y_{\tt ff}+F_{\mathtt{f}\mathsf{i}%
}(N_{\mathsf{ii}}-I)^{-1}N_{\rm i}F_{\tt f}^{\ast })= \\
\left[ 
\begin{array}{c}
B+G_{\mathsf{i}}(N_{\mathsf{ii}}-I)^{-1}N_{\mathsf{i}}G^{\ast }-\hat{A}_{%
\mathtt{sf}}\hat{Y}_{\mathtt{ff}}^{-1}\hat{A}_{\mathtt{f}} \\ 
-\hat{N}^{-1}(G_{\mathsf{e}}-G_{\mathsf{i}}(N_{\mathsf{ii}}-I)^{-1}N_{%
\mathsf{ie}})^{\ast }+\hat{N}\hat{F}_{\mathtt{f}}^{\ast }\hat{Y}_{\mathtt{ff}%
}^{-1}\hat{A}_{\mathtt{f}}
\end{array}
\right.  \\
\left. 
\begin{array}{c}
G_{\mathsf{e}}-G_{\mathsf{i}}(N_{\mathsf{ii}}-I)^{-1}N_{\mathsf{ie}}-\hat{A}%
_{\mathtt{sf}}\hat{Y}_{\mathtt{ff}}^{-1}\hat{F}_{\mathtt{f}} \\ 
\hat{N}+\hat{N}\hat{F}_{\mathtt{f}}^{\ast }\hat{Y}_{\mathtt{ff}}\hat{F}-I
\end{array}
\right] .
\end{gather*}
By direct comparison of the entries of $\mathcal{A}\mathcal{F}\mathbf{G}%
^{(k)}$ as given above with the corresponding entries of $P_{\mathtt{s}}%
\bigl((\mathbf{G}_{E}/(N_{\mathsf{ii}}-I))/(Y_{\mathtt{ff}}+F_{\mathtt{f}%
\mathsf{i}}(N_{\mathsf{ii}}-I)^{-1}N_{\mathsf{i}}F_{\mathtt{f}}^{\ast })%
\bigr)P_{\mathtt{s}}\mid _{\mathfrak{h}_{\mathtt{s}}}$, we conclude that 
\begin{equation*}
\mathcal{A}\mathcal{F}\mathbf{G}^{(k)}=P_{\mathtt{s}}\bigl((\mathbf{G}%
_{E}/(N_{\mathsf{ii}}-I))/(Y_{\mathtt{ff}}+F_{\mathtt{f}\mathsf{i}}(N_{%
\mathsf{ii}}-I)^{-1}N_{\mathsf{i}}F_{\mathtt{f}}^{\ast })\bigr)P_{\mathtt{s}%
}\mid _{\mathfrak{h}_{\mathtt{s}}}.
\end{equation*}
\end{proof}

\subsection{Commutavity of the adiabatic elimination and instantaneous
feedback operations}

We are now in a position to investigate the
commutativity of the adiabatic elimination and instantaneous feedback limit
operations for a dynamical quantum network with Markovian components. First, 
note that if $
(\mathbf{G}_{E}/Y_{\mathtt{ff}})/(N_{\mathsf{ii}}+N_{\mathsf{i}}F_{\mathtt{f}%
}^{\ast }Y_{\mathtt{ff}}^{-1}F_{\mathtt{f}\mathsf{i}}-I)= 
(\mathbf{G}_{E}/(N_{\mathsf{ii}}-I))/(Y_{\mathtt{ff}}+F_{\mathtt{f}\mathsf{i}%
}(N_{\mathsf{ii}}-I)^{-1}N_{\mathsf{i}}F_{\mathtt{f}}^{\ast })$ then $\mathcal{A}\mathcal{F}\mathbf{G}^{(k)}=\mathcal{F}\mathcal{A}\mathbf{G}%
^{(k)}$. Next, let us introduce the following notation. Let $\mathcal{I}=\{1,2,\ldots ,n\}$ and let $X$ be a $n \times n$ matrix with operator entries. For any
set of \emph{distinct} indices $\mathcal{I}_{1}=\{j_{1},j_{2},\ldots
,j_{m}\},\mathcal{I}_{2}=\{l_{1},l_{2},\ldots ,l_{m}\}\subset \mathcal{I}$ (with $m<n$) define the matrix $X_{\mathcal{I}_{1},\mathcal{I}%
_{2}}$ as $\left[ X_{jl}\right] $ with $j\in \mathcal{I}_{1}$ and $ l \in 
\mathcal{I}_{2}$. Denoting set complements as $\mathcal{I}_{1}^{c}=\mathcal{I%
}\backslash \mathcal{I}_{1}$ and $\mathcal{I}_{2}^{c}=\mathcal{I}\backslash 
\mathcal{I}_{2}$, we define the Schur complement of $X$ with
respect to a sub-matrix $X_{\mathcal{I}_{1},\mathcal{I}_{2}}$ (if it exists), denoted by $%
X/X_{\mathcal{I}_{1},\mathcal{I}_{2}}$, as 
\begin{equation*}
X/X_{\mathcal{I}_{1},\mathcal{I}_{2}}=X_{\mathcal{I}_{1}^{c},\mathcal{I}%
_{2}^{c}}-X_{\mathcal{I}_{1}^{c},\mathcal{I}_{2}}X_{\mathcal{I}_{1},\mathcal{%
I}_{2}}^{-1}X_{\mathcal{I}_{1},\mathcal{I}_{2}^{c}}.
\end{equation*}
We are now
ready to establish commutativity of successive Schur complementations, via the following lemma.

\begin{lemma}
\label{lm:successive} Let $X$ be a matrix of operators whose entries have  $\mathcal{D}$  as a
common invariant domain, and let $\mathcal{I}_{1}$, $\mathcal{I}%
_{2}$, $\mathcal{I}_{3}$ be a disjoint partitioning of the index set $%
\mathcal{I}$ of $X$ (i.e., $\cap _{j=1}^{3}\mathcal{I}_{j}=\phi $ and $\cup
_{j=1}^{3}\mathcal{I}_{j}=\mathcal{I}$). If the Schur complements 
\begin{equation*}
 X/X_{\mathcal{I}_{1}\cup \mathcal{I}_{2},\mathcal{I}_{1}\cup \mathcal{I}%
_{2}},\; 
(X/X_{\mathcal{I}_{2},\mathcal{I}_{2}})/(X/X_{\mathcal{I}_{2},\mathcal{I}%
_{2}})_{\mathcal{I}_{1},\mathcal{I}_{1}},\; \\
 (X/X_{\mathcal{I}_{1},\mathcal{I}_{1}})/(X/X_{\mathcal{I}_{1},\mathcal{I}%
_{1}})_{\mathcal{I}_{2},\mathcal{I}_{2}},
\end{equation*}
exist, then the successive Schur complementation rule holds: 
\begin{equation*}
X/X_{\mathcal{I}_{1}\cup \mathcal{I}_{2},\mathcal{I}_{1}\cup \mathcal{I}%
_{2}}=(X/X_{\mathcal{I}_{2},\mathcal{I}_{2}})/(X/X_{\mathcal{I}_{2},\mathcal{%
I}_{2}})_{\mathcal{I}_{1},\mathcal{I}_{1}}=(X/X_{\mathcal{I}_{1},\mathcal{I}%
_{1}})/(X/X_{\mathcal{I}_{1},\mathcal{I}_{1}})_{\mathcal{I}_{2},\mathcal{I}%
_{2}}.
\end{equation*}
\end{lemma}

\begin{proof}
The proof of this lemma follows \emph{mutatis mutandis} from the proof of 
\cite[Lemma 9]{GNW10} and here is somewhat simpler because the lemma
concerns ordinary Schur complements rather than generalized Schur
complements as in \cite[Lemma 9]{GNW10}. Therefore the image and kernel
inclusion conditions for  the uniqueness of the
generalized Schur complement (where the inverse is replaced by a generalized inverse) 
are not required.
\end{proof}

\begin{theorem}
\label{thm:commutativity}Under the conditions of Lemmata
\ref{lm:IF-AE} and \ref{lm:AE-IF} we have $
\mathcal{A}\mathcal{F}\mathbf{G}^{(k)}=\mathcal{F}\mathcal{A}\mathbf{G}%
^{(k)}$. Furthermore, if in addition
\begin{enumerate}
\item $\mathcal{D}$ is a core for  the operator $\mathcal{L}^{(\alpha \beta)}$ given in (\ref{cond:core}).

\item  $\mathcal{F}\mathbf{G}^{(k)}$ corresponds to a QSDE that has a unique
solution that extends to a contraction co-cycle on $\mathfrak{h}\otimes \Gamma (L_{\mathfrak{K}%
}^{2}[0,\infty ))$,

\item $\mathcal{D}$ is a core for  the operator $\mathcal{L}^{(\alpha \beta)}$ given in (\ref{cond:core}) with $\hat K, \hat L , \hat M, \hat N$ being replaced therein by the corresponding coefficients of $\mathcal{F}\mathcal{A}\mathbf{G}^{(k)}$,
\end{enumerate}
then the instantaneous feedback and adiabatic elimination operations can be
commuted. That is, applying adiabatic elimination followed by instantaneous
feedback or, conversely, applying instantaneous feedback followed by
adiabatic elimination yields the same QSDE and this QSDE has a unique
solution that extends to a unitary co-cycle on $\mathfrak{h}_{\mathtt{s}}\otimes 
\Gamma (L_{\mathfrak{K}%
}^{2}[0,\infty ))$.
\end{theorem}

\begin{proof}
If $\left[ 
\begin{array}{cc}
Y_{\mathtt{ff}} & F_{\mathtt{f}\mathsf{i}} \\ 
-N_{\mathsf{i}}F_{\mathtt{f}}^{\ast } & N_{\mathsf{ii}}-I
\end{array}
\right] $ is invertible, the Schur complement
\begin{equation*}
\mathbf{G}_{E}/\left[ 
\begin{array}{cc}
Y_{\mathtt{ff}} & F_{\mathtt{f}\mathsf{i}} \\ 
-N_{\mathsf{i}}F_{\mathtt{f}}^{\ast } & N_{\mathsf{ii}}-I
\end{array}
\right] 
\end{equation*}
is well-defined. However, since $Y_{\mathtt{ff}}$ is invertible and $N_{%
\mathsf{ii}}+N_{\mathsf{i}}F_{\mathtt{f}}^{\ast }Y_{\mathtt{ff}}^{-1}F_{%
\mathtt{f}\mathsf{i}}-I$ is also invertible by the conditions of Lemmata \ref{lm:IF-AE} and \ref{lm:AE-IF}, the matrix $\left[ 
\begin{array}{cc}
Y_{\mathtt{ff}} & F_{\mathtt{f}\mathsf{i}} \\ 
-N_{\mathsf{i}}F_{\mathtt{f}}^{\ast } & N_{\mathsf{ii}}-I
\end{array}
\right] $ is indeed invertible by the Banachiewicz matrix inversion formula
(e.g., see \cite[Section III-A]{GNW10}). The first result follows from this and Lemma \ref{lm:successive}.

Since now $\mathcal{A}\mathcal{F}\mathbf{G}^{(k)}=\mathcal{F}\mathcal{A}\mathbf{G}^{(k)}$, if the QSDEs corresponding to $\mathcal{A}\mathcal{F}\mathbf{G}^{(k)}$ and $\mathcal{F}\mathcal{A}\mathbf{G}^{(k)}$ have unique solutions that extend to
a unitary co-cycle on $\mathfrak{h}_{\mathtt{s}}\otimes \Gamma (L_{\mathfrak{K}%
}^{2}[0,\infty ))$ then they will coincide. Moreover, from this it follows by inspection that the remaining three conditions of the theorem guarantee that all the requirements of Theorem 
\ref{thm:s-conv} are met so that:
\begin{enumerate}
\item  $U^{\left( k\right) }(t)$ converges to $\hat{U}(t)$ in the sense of Theorem 
\ref{thm:s-conv}.

\item  The solution of the QSDE corresponding to $\mathcal{F}\mathbf{G}%
^{(k)} $ converges to the solution of the QSDE corresponding to $\mathcal{A} 
\mathcal{F}\mathbf{G}^{(k)}$ in the sense of Theorem \ref{thm:s-conv}.
\end{enumerate}
\end{proof}

Thus, we conclude that under the sufficient conditions for each of the sequence of 
operations $\mathcal{A}\mathcal{F}$ and $\mathcal{F}\mathcal{A}$, the two sequences of
operations are equivalent and yield the same reduced-complexity QSDE model.
This generalizes the results of \cite{GNW10} for 
quantum feedback networks with fast oscillatory components
to be eliminated. Remarkably, the
structural constraints imposed in \cite{BvHS07} to establish rigorous
adiabatic elimination results for open Markov quantum systems, originally
introduced for  considerations unrelated to the goals of this paper, play  a
crucial role in the algebra required for us to establish our results.
Exploiting these constraints, we proved that  both the instantaneous feedback limit and adiabatic elimination operations correspond to Schur complementation of a
common extended It\={o} generator matrix but with respect to different
sub-blocks of this matrix. From this we then showed that the instantaneous
feedback and adiabatic elimination operations are consistent and can be commuted once  each
sequence of operations is well-defined.

\bibliographystyle{plainnat}

\end{document}